\documentclass{IEEEtran4PSCC}
\ifCLASSINFOpdf
\else
\fi
\usepackage[cmex10]{amsmath}
\makeatletter
\let\old@ps@headings\ps@headings
\let\old@ps@IEEEtitlepagestyle\ps@IEEEtitlepagestyle
\def\psccfooter#1{%
    \def\ps@headings{%
        \old@ps@headings%
        \def\@oddfoot{\strut\hfill#1\hfill\strut}%
        \def\@evenfoot{\strut\hfill#1\hfill\strut}%
    }%
    \def\ps@IEEEtitlepagestyle{%
        \old@ps@IEEEtitlepagestyle%
        \def\@oddfoot{\strut\hfill#1\hfill\strut}%
        \def\@evenfoot{\strut\hfill#1\hfill\strut}%
    }%
    \ps@headings%
}
\makeatother

\usepackage{graphicx}      

\usepackage{commath}
\usepackage{mathtools}
\usepackage{steinmetz}		
\usepackage{supertabular}	
\usepackage{tabularx}		
\usepackage{algorithm}		
\usepackage{algorithmic}		
\usepackage{amsmath,amssymb}
\newtheorem{rem}{Remark}

\newtheorem{thm}{Theorem}
\newtheorem{ass}{Assumption}
\newtheorem{lemma}{Lemma}
\newenvironment{proof}{\vspace{0cm}\paragraph*{Proof}}{\vspace{0cm}\hfill$\blacksquare$}
\usepackage{subfig}
\usepackage{color}
\usepackage{tikz}
\usetikzlibrary{shapes,arrows}
\usepackage{soul}
\usepackage{url} 
\usepackage{afterpage}
\usepackage[makeroom]{cancel}
\usepackage{enumitem} 
\usepackage[framemethod=tikz]{mdframed}
\usepackage{framed}
\colorlet{shadecolor}{white}
\usepackage{cite}

\newcommand{\normsz}[1]{\lVert #1 \rVert_2}
\newcommand{\expect}[1]{\mathbb{E}\big[ #1 \big]}

\renewcommand{\r}{\textcolor{black}}

\usepackage[font=small]{caption}

\begin{document}
%
\title{Closing the Loop: Dynamic State Estimation and Feedback Optimization of Power Grids}

\author{
\IEEEauthorblockN{Miguel Picallo, Saverio Bolognani, Florian D{\"o}rfler}
\IEEEauthorblockA{Automatic Control Laboratory, ETH Zurich, 8092 Zurich, Switzerland \\ \{miguelp,bsaverio,dorfler\}@ethz.ch
}
}



\maketitle

\begin{abstract}
This paper considers the problem of online feedback optimization to solve the AC Optimal Power Flow in real-time in power grids. This consists in continuously driving the controllable power injections and loads towards the optimal set-points in time-varying conditions based on real-time measurements performed on the grid. However, instead of assuming noise-free full state measurement like in recently proposed feedback optimization schemes, we connect a dynamic State Estimation using available measurements, and study its dynamic interaction with the optimization scheme. We certify stability of this interconnection and the convergence in expectation of the state estimate and the control inputs towards the true state values and optimal set-points respectively. Additionally, we bound the resulting stochastic error. Finally, we show the effectiveness of the approach on a test case using high resolution consumption data.
\end{abstract}


\thanksto{Funding by the Swiss Federal Office of Energy through the projects ''Renewable Management and Real-Time Control Platform (ReMaP)'' (SI/501810-01) and ''A Unified Control Framework for Real-Time Power System Operation (UNICORN)'' (SI/501708) and by the ETH Foundation is gratefully acknowledged.}

\section{Introduction}

The operation of power grids, and especially distribution grids, is undergoing a paradigm shift due to the increasing share of controllable elements (generation power, curtailment, reactive power in converters, flexible loads, etc.), and pervasive sensing (smart meters, phasor measurement units, etc.). Moreover, the introduction of communication networks with high resolution sensor sampling, and the fast response offered by the power-electronics interfacing the controllable energy sources, enable to perform very fast control-loop rates. These new technologies offer the potential advantage of lowering the grid operation cost, promoting more sustainable energy systems, improving reliability and enabling a more efficient use of the existing infrastructure. Nonetheless, the unpredictability and the high variability of household loads and renewable energy sources, pose a severe challenge in satisfying the required grid specifications, like voltage levels, thermal limits, etc. Consequently, these specifications need to be enforced through either offline optimization or real-time feedback control.

State-of-the-art optimization methods typically solve a static Optimal Power Flow (OPF) to determine the set-points of these controllable energy sources \cite{molzahn2017survey}. However, these approaches may not succeed in efficiently controlling the system under fast time-varying conditions, since they do not take full advantage of these fast sensor sampling and control-loop rates. Additionally, they rely on accurate grid models and measurements which are seldom available. On the other hand, the recently proposed online feedback optimization \cite{bolognani2015reactivePF, anese2016optimal, tangrtopf2017, hauswirth2017onlinePF, colombino2019robustness} has shown an outstanding performance for real-time power system operation under variable conditions and safety requirements. This online feedback optimization consist of collecting real-time grid state measurements at each time step, and then use them as feedback to a controller that incrementally drives the controllable power injections towards the optimal set-points. It has been been shown, that this feedback optimization offers the advantages of quickly adapting to time-varying conditions \cite{anese2016optimal} and an improved robustness against model-mismatch \cite{colombino2019robustness}.

However, so far the analysis of all of these feedback optimization approaches have considered a stylized problem setup by assuming availability of noise-free measurements of all states that need to be controlled, potentially the entire system state. Yet, state measurements may be scarce or present disturbance noises. Such is the case in distribution grids, which typically present heterogeneous measurements and are not observable, and thus require inaccurate a-priori information in the form of load predictions, also known as pseudomeasurements, to achieve observability \cite{schenato2014bayesian, picallo2017twostepSE}. Such a setup calls for a State Estimation (SE) \cite{abur2004power}, which in fact has been considered for the standard offline OPF problem \cite{picallo2018stochOPF}, but not yet for the online feedback optimization. There is one main obstacle on the way: even if the SE and the feedback optimization are both separately stable and optimal, this does not guarantee that their interconnection will inherit these properties. 

Therefore, in this paper, we combine a dynamic SE \r{\cite{zhao2019revsdynse}}, based on a Kalman filter, and the online feedback optimization. The contribution of this work lies in formally proving that the interconnection of the grid dynamics, the SE, and the feedback optimization is stable and steady-state optimal. More concretely, we certify that in the presence of process and measurement noise, the state estimate and the power set-points delivered by our method converge in expectation to the true state and the optimal set-points, respectively, and both have a bounded error covariance. Additionally, we show the effectiveness of our approach in the IEEE 123-bus test feeder \cite{kersting1991radial} with highly uncertain pseudo-measurements and high resolution consumption data.

The rest of the paper is structured as follows: Section~\ref{sec:prelim} introduces some preliminaries: the distribution grid model, the problem setup and its linearization, the online feedback optimization, and the measurements. Section~\ref{sec:SEoptim} presents the proposed method combining SE and feedback optimization, and proves its stability and convergence. Finally, in Section~\ref{sec:sim} the approach is validated on a simulated test feeder. The proof of the main result is in Appendix \ref{app:thmproof}. 

\section{Preliminaries}\label{sec:prelim}

\subsection{Distribution Grid Model}\label{sec:grid}
A distribution grid can be modelled as a graph $\mathcal{G}=(\mathcal{V},\mathcal{E},\mathcal{W})$ with nodes $\mathcal{V}=\{1,...,N_\text{bus}\}$ representing the buses, edges $\mathcal{E}=\{(v_i,v_j)  \mid \hspace{-0.05cm} v_i,v_j \in \mathcal{V}\}$ representing the branches, and edge weights $\mathcal{W}=\{w_{i,j} \hspace{-0.05cm} \mid \hspace{-0.05cm} (v_i,v_j) \hspace{-0.05cm} \in \hspace{-0.05cm} \mathcal{E}, w_{i,j} \hspace{-0.05cm} \in \hspace{-0.05cm} \mathbb{C}\}$ representing the admittance of a branch.

In 3-phase networks buses may have up to 3 phases, so that the voltage at bus $i$, with $n_{\phi,i}\leq 3$ phases, is $U_{\text{bus},i} \in \mathbb{C}^{n_{\phi,i}}$ (and the edge weights $w_{i,j}\in \mathbb{C}^{n_{\phi,i} \times n_{\phi,j}}$). The state of the network is then represented by the vector bus voltages $U_\text{bus}=[U_\text{pcc}^T, \; U^T]^T \in \mathbb{C}^{N+3}$, where $U_{\text{pcc}} \in \mathbb{C}^3$ denotes the measured voltage at the point of common coupling (PCC) connected to the main grid, and $U \in \mathbb{C}^N$ are the voltages in the non-source buses, where $N$ depends on the number of buses and phases per bus. The Laplacian matrix $Y \in \mathbb{C}^{(N+3) \times (N+3)}$ of the weighted graph $\mathcal{G}$, also called admittance matrix, can be used to express the power flow equations that relate the currents $I \in \mathbb{C}^N$ and the active and reactive power injections $P,Q \in \mathbb{R}^N$ at each node:
\begin{equation}\label{eq:PFeq}\arraycolsep=1pt
\begin{array}{c}
\left[\begin{array}{c} I_{\text{pcc}} \\ I \end{array}\right] = 
Y\left[\begin{array}{c} U_{\text{pcc}} \\ U \end{array}\right] = \
\left[\begin{array}{cc} Y_{00} & Y_{01} \\ Y_{01}^T & Y_{11} \end{array}\right]
\left[\begin{array}{c} U_{\text{pcc}} \\ U \end{array}\right]\\[0.4cm]

P+jQ= \text{diag}(U)\bar{I} 
= \text{diag}(U)(
 \bar{Y}_{01}^T\bar{U}_{\text{pcc}} + \bar{Y}_{11}\bar{U}),
\end{array}
\end{equation}
where $j$ is the imaginary unit, $\bar{(\cdot)}$ denotes the complex conjugate, and $\text{diag}(\cdot)$ represents the diagonal operator, converting a vector into a diagonal matrix. 

Moreover, within the set of nodes $\mathcal{V}$ we distinguish the set of nodes with controllable power injections $\mathcal{C}$, and uncontrollable loads $\mathcal{L}$; and define the vectors $S_c , S_l $ with the corresponding power values. Thus, we can split $P,Q$ as follows:
\begin{equation*}\arraycolsep=1pt
\begin{array}{c}
\left[\begin{array}{c} P \\ Q \end{array}\right] = 
\mathbb{I}_c S_{c} + \mathbb{I}_l S_{l} , \hspace{0.1cm}
S_c \hspace{-0.05cm} = \hspace{-0.05cm} \left[\begin{array}{c} P_c \\ Q_c \end{array}\right] \hspace{-0.05cm} \in \mathbb{R}^{2N_c}, 
S_l \hspace{-0.05cm} = \hspace{-0.05cm} \left[\begin{array}{c} P_l \\ Q_l \end{array}\right] \hspace{-0.05cm} \in \mathbb{R}^{2N_l},
\end{array}
\end{equation*}
where $\mathbb{I}_c,\mathbb{I}_l$  are matrices filled with $0,1$ that link the elements of $S_c, S_l$ in the sets $\mathcal{C},\mathcal{L}$ to the corresponding nodes indices of $[P^T,Q^T]^T$, and $N_c,N_l$ are the number of elements in $\mathcal{C},\mathcal{L}$. 

\subsection{Problem Setup}

The operation of these distribution grids consists in optimizing the use of the controllable resources $S_c$, while satisfying some restrictions on the grid state $U$, which can be represented using polar coordinates as: $V=[\: \abs{U}^T,\angle{U}^T]^T \in \mathbb{R}^{2N}$. Then this optimization problem can be expressed as a particular case of the AC Optimal Power Flow (AC-OPF):
\begin{equation}\label{eq:opt0}
\min_{S_{c} \in \mathcal{F},V} f(S_{c}) + g(V) \text{ s.t. } \eqref{eq:PFeq},
\end{equation}
where 
\begin{itemize}[leftmargin = 1cm]
\item[$\mathcal{F}$:] The feasible set $\mathcal{F}$ contains the limits of available power, for example $\mathcal{F}=\{S_c | P_{i,\min} \leq P_{c,i} \leq P_{i,\max}, \; Q_{i,\min} \leq Q_{c,i} \leq Q_{i,\max}, \; \forall i=1,\dots,N_c\}$. These are hard constraints to be satisfied all the time.
\item[$f(S_{c})$:] The objective function $f(S_c)=\sum_{i=1}^{N_c} f_i(S_{c,i})$ determines the cost of the grid operations by adding the costs $f_i$ of using each controllable power $i$. 
\item[$g(V)$:] The objective function $g(V)$ encodes the grid restrictions through the specification of desired states and objectives. These can include power losses, voltage constraints, overload of lines and transformers, etc. \r{Similar to \cite{hauswirth2017onlinePF}}, we will use $g(V)=\frac{\rho}{2}\sum_{i=1}^N (\max(0,\abs{U}_i-V_{\max}))^2 + (\max(0,V_{\min}-\abs{U}_i))^2$ to penalize violations of the voltage limits. These are soft constraints that can be violated at a given cost, especially during transients. However, we can enforce fewer and lower violations if using large values for $\rho$. \r{Alternatively, these restrictions could be enforced as hard constraints \cite{anese2016optimal}. For that, we would consider the dual variables in the corresponding Lagrangian, and optimize over them alongside the primal variables.}
\end{itemize}

\subsection{Power Flow Linear Approximation}

The non-linearity of the power flow equations \eqref{eq:PFeq} is difficult to handle when solving the optimization problem \eqref{eq:opt0}. For multi-phase unbalanced distribution grids, these equations \eqref{eq:PFeq} are generally nonconvex \cite{low2014exactness}. Hence, \r{linear approximations are used frequently in the literature \cite{bolognani2016existence, anese2016optimal, molzahn2019surveyrel}}. Considering the flat voltage (zero injection profile) $U_0=Y_{11}^{-1} Y_{01}^T U_{\text{pcc}}$ ($V_0=[\: \abs{U_0}^T,\angle{U_0}^T]^T$) as operating point, and defining $Z_{V} = Y_{11}^{-1}\text{diag}(\bar{U_0})^{-1}$, we have the linearisation
\begin{equation}\label{eq:PFapprox}\arraycolsep=1pt\begin{array}{c}
V = V_0 + B_cS_c + B_l S_l \\[0.2cm] 
B_c=B\mathbb{I}_c, B_l=B\mathbb{I}_l \\[0.1cm]
B=
\left[\begin{array}{cc} \cos(\angle{U_0}) & \sin(\angle{U_0}) \\ \frac{-\sin(\angle{U_0})}{\abs{U_0}} & \frac{\cos(\angle{U_0})}{\abs{U_0}} \end{array}\right]
\left[\begin{array}{cc} \Re\{Z_{V}\} & \Im\{Z_{V}\} \\ \Im\{Z_{V}\} & -\Re\{Z_{V}\} \end{array}\right],
\end{array}
\end{equation}
where $\Re\{\cdot\},\Im\{\cdot\}$ denote the real and imaginary part of a complex number.

Instead of \eqref{eq:opt0}, we consider then its linear approximation:
\begin{equation}\label{eq:opt}
\min_{S_c \in \mathcal{F},V} f(S_c) + g(V) \text{ s.t. }  \eqref{eq:PFapprox}.
\end{equation}

\subsection{Online Feedback Optimization}

The static optimization problem \eqref{eq:opt0} is, in essence, a \textit{feedforward} controller that computes the optimum set-points given the estimated conditions and grid model. However, these set-points become suboptimal or outdated given fast time-varying conditions, model-mismatches and estimation errors \cite{doerfler2019autonomousgrids}. 

To mitigate these effects, the online feedback optimization consist of using the current grid state $V$ as \textit{feedback} to a controller that drives the controllable power towards the optimal set-points in real-time. Typically, these controllers use a variation of projected gradient descent \cite{hauswirth2017onlinePF,colombino2019robustness}. The gradient brings the set-points closer to the optimal ones, while the projection forces the solution to remain within the feasible set $\mathcal{F}$. At every time-step from $(t)$ to $(t+1)$ we update
\begin{equation}\label{eq:projgrad0}
S_{c,(t+1)} = \Pi_{\mathcal{F}} \big[S_{c,(t)} - \epsilon \big(\nabla_{S_c} f(S_{c(t)}) + B_c^T \nabla_V g(V_{(t)}) \big) \big],
\end{equation}
where $\epsilon$ is the descent rate of the gradient method, $\nabla_{S_c} f(S_c)$ and $\nabla_V g(V)$ are the gradients of $f$ and $g$ respectively, and $\Pi_{\mathcal{F}}[\cdot]$ denotes the projection on the feasible set $\mathcal{F}$ ($\Pi_{\mathcal{F}}[x]=\arg\min_{z \in \mathcal{F}}\normsz{z-x}$). This expression is derived by computing the gradient of the objective function $f(S_c)+g(V)$ with respect to the decision variables $S_c$: $\nabla_{S_c} (f(S_c)+g(V))= \nabla_{S_c}f(S_c) + \big(\frac{\partial V}{\partial S_c}\big)^T \nabla_V g(V)$, where $\frac{\partial V}{\partial S_c}=B_c$. 

\begin{rem}\label{rem:constfeas}
Note that in \eqref{eq:projgrad0} we are assuming that the feasible set $\mathcal{F}$ from \eqref{eq:opt0} and \eqref{eq:opt} remains constant and independent of time. In a real-world application, it would be time-varying, since the amount of power available may change over time, for example if delivered by renewable sources. However, when considering fast time scales, the change of $\mathcal{F}$ at subsequent instants tends to $0$ and can be neglected. We will observe this later in the test case in Section \ref{sec:sim}.
In any case, this work could be extended to a time-varying feasible set $\mathcal{F}_{(t)}$ as in \cite{tangrtopf2017}.
\end{rem}

Apart from the control inputs $S_c$, which are measurable, the state-of-the-art online feedback optimization approaches \eqref{eq:projgrad0} assume a noise-free deterministic knowledge of the full state $V$ to evaluate $\nabla_V g(V_{(t)})$ \cite{tangrtopf2017, hauswirth2017onlinePF}. Other approaches do not require full state measurements, but restrict $g$ to control only the subset of directly measured states \cite{bolognani2015reactivePF, anese2016optimal,colombino2019robustness}. 

\subsection{Measurements}\label{sec:meas}

Since \eqref{eq:projgrad0} requires $V$, and similarly the optimization problems \eqref{eq:opt0},\eqref{eq:opt} need $S_l$, neither can be solved without knowing $V$ or $S_{l}$. Therefore, we need to collect measurements. In distribution grids, there may be a set of heterogeneous noisy measurements, like voltage, current and power magnitudes and/or phasors. However, due to the scarce number of measurements, they may need to be complemented with the so-called pseudo-measurements \cite{schenato2014bayesian, picallo2017twostepSE}, i.e., low accuracy load forecasts, to make the system numerically observable \cite{baldwin1993power}, and thus be able to solve a SE problem. For simplicity, in this paper we will assume that we have a linearised measurement equation that contains the measurements coming from all sources (e.g. conventional remote terminal units, smart meters, phasor measurement units, pseudomeasurements, etc.):
\begin{equation}\label{eq:meas}
y = HV + \omega_y,
\end{equation}
where $y$ is the vector of measurements, $H$ is the matrix mapping the state to the measurements, and $\omega_y$ is the measurement noise. We assume that this noise is Gaussian with known probability distribution $\omega_y \sim \mathcal{N}(0,\Sigma_y)$, and that using the pseudo-measurements, the matrix $H$ has full-column rank, and thus the system is numerically observable \cite{baldwin1993power}. 

Since these pseudo-measurements can be simply load profiles for different kinds of nominal consumption (household, office building, etc.), and thus have a low accuracy, we will assign relatively large values to the corresponding covariance terms in $\Sigma_y$, indicating a large uncertainty \cite{picallo2017twostepSE}. 

\section{State Estimation for Feedback Optimization}\label{sec:SEoptim}

Instead of the standard online feedback optimization \eqref{eq:projgrad0}, in our approach we consider a more general and realistic scenario, where the whole state $V$ needs to be controlled, but neither the loads $S_l$ nor the full state $V$ are directly measured. To retrieve this information, we use the available measurements $y$ from \eqref{eq:meas} to build a dynamic SE. Then, we connect this SE as feedback to the controller of the online feedback optimization.

\subsection{Dynamic State Estimation}\label{sec:SEdyn}

Even though we are optimizing the steady-state power flow solution represented in \eqref{eq:PFeq} and \eqref{eq:PFapprox}, we consider the stochastic dynamic system induced on the grid state $V$ by a change in the controllable injected power $S_c$ and the uncontrollable stochastic loads $S_l$. We build this system by subtracting the power flow equations \eqref{eq:PFapprox} at subsequent times $(t)$:
\begin{equation}\label{eq:dyneq}\begin{array}{l}
V_{(t)} = V_{(t-1)} + B_c (S_{c,(t)}-S_{c,(t-1)}) + \omega_{l,(t)},
\end{array}
\end{equation}
where $\omega_{l,(t)}=B_l(S_{l,(t)}-S_{l,(t-1)})$ appears as a result of the time-varying load conditions $S_{l,(t)}$. This dynamic approach \eqref{eq:dyneq} allows to circumvent the lack of precise knowledge of $S_{l,(t)}$ by considering only its time variations as process noise: $\omega_{l,(t)} \sim \mathcal{N}(0,\Sigma_l)$. Using \eqref{eq:dyneq}, we design a Kalman filter based SE \cite{monticelli2000electric}, that at time $(t)$ takes the measurements $y_{(t)}$ as input and outputs the estimate $\hat{V}_{(t)}$:
\begin{equation}\label{eq:kalmanest}\arraycolsep=1pt
\begin{array}{rl}
\hat{V}_{(t)} = &  (I_d - K_{(t)} H)\big(\hat{V}_{(t-1)} +  B_c ( S_{c,(t)} - S_{c,(t-1)} ) \big) \\ & +  K_{(t)}y_{(t)}\\[0.1cm]
P_{(t)} = & (I_d - K_{(t)}H)(P_{(t-1)}+\Sigma_l) \\[0.1cm]
K_{(t)} = & (P_{(t-1)}+\Sigma_l)H^T\big(H(P_{(t-1)}+\Sigma_l)H^T + \Sigma_y\big)^{-1}, 


\end{array}
\end{equation}  
where $I_d$ is the identity matrix, $P_{(t)}$ denotes the covariance matrix of the voltage state estimate $\hat{V}_{(t)}$, and $K_{(t)}$ is the Kalman gain matrix minimizing the resulting covariance $P_{(t)}$: $K_{(t)}=\arg\min_K \text{trace}(P_{(t)})$.

\begin{rem}\label{rem:mean0}
Note that we are assuming that the noise $\omega_{l,(t)}$ in \eqref{eq:dyneq} has $0$ mean. This is not necessary true in practice, since the loads could drift in expectation depending on the hour of the day. However, similarly as in Remark \ref{rem:constfeas}, on the considered fast time scales, this drift tends to $0$ and can be neglected. We will observe this later in the test case in Section \ref{sec:sim}. Nevertheless, this work could be extended to a more general case with non-zero mean and time dependent noise parameters: $\omega_{l,(t)} \sim \mathcal{N}(\mu_{l,(t)},\Sigma_{l,(t)})$. \r{A load prediction could be used to estimate the drift term $\mu_{l,(t)}$, and then compensate it in the optimization step}.
\end{rem}

\subsection{Convergence of the Projected Gradient Descent}

\tikzstyle{block} = [draw, fill=white, rectangle, 
    minimum height=2em, minimum width=3em]
\tikzstyle{block2} = [draw, fill=None, rectangle, 
    minimum height=2em, minimum width=3em]
\tikzstyle{sum} = [draw, fill=white, circle, node distance=1cm]
\tikzstyle{input} = [coordinate]
\tikzstyle{output} = [coordinate]
\tikzstyle{pinstyle} = [pin edge={to-,thin,black}]

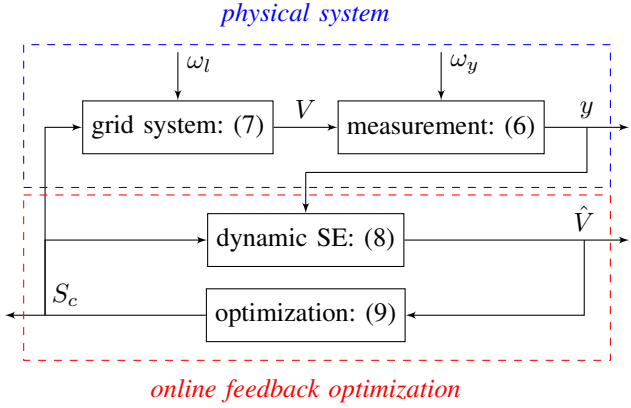
\begin{figure}[t]
\centering
\begin{tikzpicture}[auto, node distance=2cm,>=latex']
	\node [input] (input0) {};   
    \node [block, right of = input0, node distance = 2.3cm] (system) {grid system: \eqref{eq:dyneq}};
    \node [block, right of = system, node distance = 3.5cm] (meas) {measurement: \eqref{eq:meas}};
    \draw [->] (system) -- node [above,name=V] {$V$} (meas); 
    \node [output, right of = meas, node distance=2.5cm] (output0) {}; 
    \draw [->] (meas) -- node [name=y] {$y$}(output0);

	\node [input, below of = input0, node distance = 1.5cm] (input1) {};    
	\node [block, right of = input1, node distance = 4cm] (SE) { dynamic SE: \eqref{eq:kalmanest}};
	\node [input, above of = SE, node distance = 0.9cm] (inputm) {};
	\draw [-] (y) |- (inputm);	
	\draw [->] (inputm) -| (SE);	
	\node [output, below of = output0, node distance = 1.5cm] (output1) {};
	\draw [->] (SE) -- node [above,name=Vest,pos = 0.8 ] {$\hat{V}$} (output1); 
	
    \node [input, below of = input1, node distance = 1.0cm] (input2) {};
    \node [block, right of=input2, node distance=4.0cm] (grad) {optimization: \eqref{eq:projgrad}};
    \node [output, below of = output1, node distance = 1.0cm] (output2) {};
    \draw [->] (Vest) |- (grad);
    \draw [->] (grad) -- node [above,name=u,pos = 0.7] {$S_c$} 
    node [below,name=upos,pos = 0.8] {} (input2);
    \draw [->] (upos) |- (system);
    \draw [->] (upos) |- (SE);

   	\node [input, above of = system, node distance = 1cm] (distd) {};
   	\node [input, above of = meas, node distance = 1cm] (disty) {};
   	\draw [->] (distd) -- node [near start, name=wd, right] {$\omega_l$} (system);
   	\draw [->] (disty) -- node [near start, name=wy, right] {$\omega_y$} (meas);

   	\draw [red,dashed] ([yshift=0.6cm,xshift=0.25cm]input1) -- ([yshift=0.6cm,xshift=-0.25cm]output1) -- ([yshift=-0.6cm,xshift=-0.25cm]output2) -- ([yshift=-0.6cm,xshift=0.25cm]input2) -- ([yshift=0.6cm,xshift=0.25cm]input1);
   	\node[text=red, below of = grad, node distance = 1cm] (feedback) {\textit{online feedback optimization}};

   	\draw [blue,dashed] ([yshift=0.7cm,xshift=0.25cm]input1) -- ([yshift=0.7cm,xshift=-0.25cm]output1) -- ([yshift=1.1cm,xshift=-0.25cm]output0) -- ([yshift=1.1cm,xshift=0.25cm]input0) -- ([yshift=0.7cm,xshift=0.25cm]input1);
   	\node[text=blue, above of = grad, node distance = 4cm] (system) {\textit{physical system}};

\end{tikzpicture}

\caption{Block diagram of the interconnected systems: distribution grid system \eqref{eq:dyneq}, measurements \eqref{eq:meas}, state estimation \eqref{eq:kalmanest} and online feedback optimization \eqref{eq:projgrad}. 
}\label{fig:bldiag}
\end{figure}

Instead of using the state $V_{(t)}$ as in \eqref{eq:projgrad0}, we use $\hat{V}_{(t)}$ from \eqref{eq:kalmanest} as feedback to the projected gradient descent:
\begin{equation}\label{eq:projgrad}
S_{c,(t+1)} = \Pi_{\mathcal{F}} \big[ S_{c,(t)} - \epsilon \big(\nabla_{S_c} f(S_{c,(t)}) + B_c^T \nabla_V g(\hat{V}_{(t)}) \big) \big].
\end{equation} 

The interconnection of these subsystems \eqref{eq:kalmanest} and \eqref{eq:projgrad} with the stochastic dynamic system of the grid \eqref{eq:dyneq} and the measurement equation \eqref{eq:meas}, results in the closed-loop system represented in Figure~\ref{fig:bldiag}. However, even if the SE \eqref{eq:kalmanest} converges in expectation to unbiased estimate with finite variance, and the online feedback optimization \eqref{eq:projgrad}, converges asymptotically to the solution of the OPF \eqref{eq:opt}; this does not guarantee that their interconnection will inherit these properties. Therefore, we need to verify the overall stability, and provide the corresponding convergence rates, to ensure the desired behaviour of our approach. 

\begin{ass}\label{ass:thm}
To prove this stability and convergence, we need the following technical assumptions:
\begin{enumerate}[leftmargin=*]
\item Both functions $f(\cdot),g(\cdot)$ have Lipschitz continuous gradients with parameters $L_f,L_g$ respectively (i.e. $\normsz{\nabla_{S_c} f(S_{c,1})-\nabla_{S_c} f(S_{c,2})} \leq L_f \normsz{S_{c,1}-S_{c,2}}$, $\normsz{\nabla_{V} g(V_{1})-\nabla_{V} g(V_{2})} \leq L_g \normsz{V_{1}-V_{2})}$).
\item The function $g(\cdot)$ is convex.
\item The function $f(\cdot)$ is $\eta$-strongly-convex (i.e. $(\nabla_{S_{c}} f(S_{c,1})-\nabla_{S_{c}} f(S_{c,2}))^T(S_{c,1}-S_{c,2}) \geq \eta \normsz{S_{c,1}-S_{c,2}}^2$). 
\item The measurements matrix $H$ has full-column rank.
\item The noise covariance matrices $\Sigma_y,\Sigma_l$ are bounded. 
\item The noise covariance matrix $\Sigma_y$ has full rank. 
\item \label{assi:lbPt} The covariance matrices $P_{(t)}$ are lower bounded (i.e., there exists a parameter $\check{\sigma}>0$ so that $P_{(t)} \succeq \check{\sigma}I_d \; \forall (t)$).
\end{enumerate}
\end{ass}

These assumptions are standard and relatively easy to satisfy: The Lipschitz continuity assumption is typically true for most functions used in these kind of applications. The convexity of $g(\cdot)$ will usually hold if using penalization functions like the one described after \eqref{eq:opt0}. Note that strict convexity is not required for $g(\cdot)$. If necessary, the strong convexity of $f(\cdot)$ could always be satisfied by adding some quadratic regularization term to the objective function, like in \cite{anese2016optimal}. This is a standard procedure to achieve well-posed problems. The full-column rank of $H$ is easy to achieve if including pseudo-measurements, see Section~\ref{sec:meas}. Since in power grids physical quantities are typically bounded, so will be the noise covariance matrices $\Sigma_y,\Sigma_l$. A full rank covariance matrix $\Sigma_y$ is to be expected, since measurement noises originate in different sensors. A sufficient condition to have a lower bounded $P_{(t)}$ is a full rank $\Sigma_l$, and thus a uniformly completely controllable system \cite[Lemma~7.2]{jazwinski1970stochfilt}, but this is not necessarily true. Even after using a Kron reduction to eliminate zero-injection nodes, there may be nodes with controllable deterministic power injection, but no stochastic input. Then $B_l$ is not full-row rank, and thus $\Sigma_l$ is not full rank, since $\omega_{l,(t)}=B_l(S_{l,(t)}-S_{l,(t-1)})$. However, it would be possible to add some stochastic uncertainty in this nodes, like fictitious low loads or power injection noises, so that $[B_c,B_l]$ becomes a full-row rank matrix, and $\Sigma_l$ full rank.

Given the convexity of $f(\cdot),g(\cdot)$, we can then define the instantaneous global optimal value of \eqref{eq:opt} at time $(t)$: $S_{c,(t)}^*$, depending on the stochastic realization of $S_{l,(t)}$ at time $(t)$; and $S_{c}^{\mathbb{E}*}$ the expected global optimal value using the expected loads $\expect{S_{l,(t)}}$ in \eqref{eq:opt}. Since $\expect{\omega_{l,(t)}}=0$, $\expect{S_{l,(t)}}$ is constant in time $(t)$ and so is $S_{c}^{\mathbb{E}*}$.

\begin{thm}\label{thm:conv}
Under Assumption \ref{ass:thm} and choosing for the projected gradient descent \eqref{eq:projgrad} a descent rate $\epsilon < \frac{2\eta}{(L_f+\normsz{B_c}^2 L_g)^2+\normsz{B_c}^2 L_g^2}$, the system in Figure~\ref{fig:bldiag} has the following stability and convergence results:
\begin{itemize}[leftmargin=*]
\item Exponential convergence towards the unbiased and optimal solution: there exist constants $C_{V,1},C_{V,2},C_{S,1},C_{V,2} > 0$ such that 
\begin{equation}\label{eq:convunbiassol}\begin{array}{l}
\normsz{\expect{\hat{V}_{(t)} \hspace{-0.05cm} - \hspace{-0.05cm} V_{(t)}}}^2 
\leq C_{V,1} e^{-C_{V,2}t}\normsz{\expect{\hat{V}_{(0)} \hspace{-0.05cm} - \hspace{-0.05cm} V_{(0)}}}^2 \stackrel{t \rightarrow \infty}{\rightarrow} 0 \\[0.2cm] 

\normsz{S_{c,(t)}-S_{c}^{\mathbb{E}*}} 
 \leq t C_{S,1} e^{-C_{S,2}t} \normsz{S_{c,(0)}-S_{c,(0)}^{\mathbb{E}*}} \stackrel{t \rightarrow \infty}{\rightarrow} 0
\end{array}
\end{equation}
\item Exponential bounded mean square (stochastic stability): there exist constants $C_V,C_{V,3},C_{V,4},C_{S},C_{S,3},C_{V,4} > 0$ such that
\begin{equation}\label{eq:convboundedmsq}\begin{array}{l}
\expect{\normsz{\hat{V}_{(t)}-V_{(t)}}^2} \\
\leq C_{V} + C_{V,3}e^{-C_{V,4}t}\expect{\normsz{\hat{V}_{(0)}-V_{(0)}}^2}
\stackrel{t \rightarrow \infty}{\rightarrow} C_{V} < \infty\\[0.2cm]

\expect{\normsz{S_{c,(t)} \hspace{-0.05cm} - \hspace{-0.05cm} S_{c,(t)}^*}} \\
\leq C_{S} + tC_{S,3}e^{-C_{S,4}t}\expect{\normsz{S_{c,(0)} \hspace{-0.05cm} - \hspace{-0.05cm} S_{c,(0)}^*}} \stackrel{t \rightarrow \infty}{\rightarrow} C_{S}  < \infty
\end{array}
\end{equation}
\end{itemize}
\end{thm}

Both error terms $\hat{V}_{(t)}-V_{(t)}$ and $S_{c,(t)}-S_{c,(t)}^*$ are stochastic processes that quantify the estimation and the optimality errors respectively. The first result \eqref{eq:convunbiassol} establishes that the expected values of these errors converge towards $0$, while the second result \eqref{eq:convboundedmsq} bounds their covariances. Both are required to conclude that the closed-loop stochastic dynamic system in Figure~\ref{fig:bldiag} converges and is stable. Moreover, given \eqref{eq:PFapprox}, $V_{(t)}-V_{(t)}^* = B_c(S_{c,(t)}-S_{c,(t)}^*)$, so the convergence of $V_{(t)}$ is a result of the convergence of $S_{c,(t)}$. Note that the convergence of $\hat{V}_{(t)}-V_{(t)}$ is faster than the one of $S_{c,(t)}-S_{c,(t)}^*$ due to the $t$ multiplying the exponential in the later case.

These constants $C_{V},C_{S},C_{V,(\cdot)},C_{S,(\cdot)}$ depend on the parameters $\eta,L_f,L_g,B_c,B_l,H,\check{\sigma},\Sigma_l,\Sigma_y$. Their expressions can be found in the proof in Appendix \ref{app:thmproof}. Moreover, $C_{V},C_{S}$ are monotone increasing with respect to $\text{trace}(\Sigma_l),\text{trace}(\Sigma_y)$. 
So they could be lowered if using faster time scales, since the loads would be expected to change less and thus $\text{trace}(\Sigma_l)$ would be lower.

\section{Test Case}\label{sec:sim}

We validate our proposed method by simulating the behaviour of a test distribution grid during $30$ minutes with $1$-second time intervals (1800 iterations). 

\subsection{Settings}\label{subsec:simDef}

\begin{figure}
\centering
\includegraphics[width=8cm,height=6.4cm]{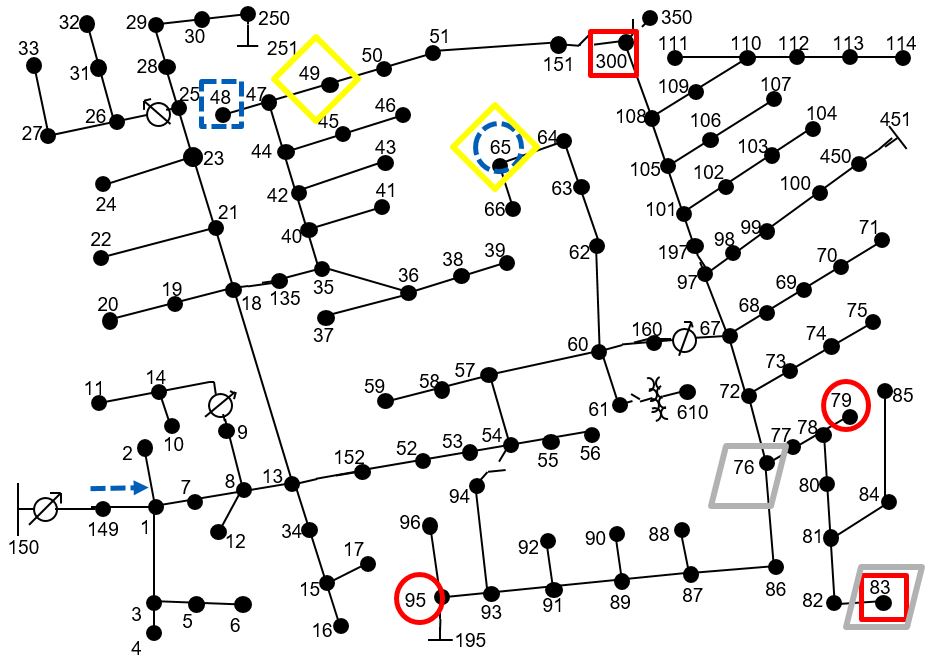}    
\caption{IEEE 123-bus test feeder \cite{kersting1991radial}. \textbf{Measurements:} red circle~=~voltage phasor, red square~=~voltage magnitude, blue dashed circle~=~current phasor, blue dashed square~=~current magnitude, blue dashed arrow~=~line current phasor. \textbf{Distributed generation:} yellow diamond~=~solar, grey parallelogram~=~wind.} 
\label{fig:123bus}
\end{figure}

\noindent Simulation data settings:

\begin{itemize}
\item System: 3-phase, unbalanced IEEE 123-bus test feeder \cite{kersting1991radial} (see Fig.~\ref{fig:123bus}). 

\item Load Profiles: $1$-second resolution data of the ECO data set \cite{ECOdata}. To adapt the active and reactive loads to this grid, we aggregate households and rescale them to the base loads of the 123-bus feeder.

\item Measurements for the SE (see Figure~\ref{fig:123bus}): Voltage measurements are placed at buses $95,79,300$ and $83$, current measurements at buses $65$ and $48$, and branch current phasor measurements at branch $149$-$1$ after the regulator.  As in \cite{picallo2017twostepSE}, we assign a relatively low value (1\% standard deviation) to the covariance terms in $\Sigma_y$ corresponding to these measurements.

\item Pseudo-measurements: We build these profiles by averaging the load profiles in the ECO data set \cite{ECOdata}. As in \cite{picallo2017twostepSE}, we assign a relatively high value (50\% standard deviation) to the corresponding covariance terms in $\Sigma_y$. 

\item Distributed Generation: Similar to \cite{picallo2018stochOPF}, solar energy is introduced in the three phases of nodes $49$ and $65$, and wind energy in nodes $76$ and $83$, see Figure~\ref{fig:123bus}. Their profiles are simulated using a $1$-minute solar irradiation profile and a $2$-minute wind speed profiles from \cite{solarprofile, windprofile}. Generation is assumed constant between samples.
\end{itemize}

Optimization settings:
\begin{itemize}
\item Objective function $f(S_c)$: We consider a quadratic cost on the controllable resources that penalises not using all the available active power in the renewable resources, and also using any reactive power: $f_i((P,Q)_{c,i})=\frac{1}{2}(P_{c,i}-P_{i,\max})^T(P_{c,i}-P_{i,\max})+\frac{1}{2}Q_{c,i}^TQ_{c,i}$. This function is strongly convex and has a Lipschitz continuous gradient, with parameters $\eta = L_f=1$. 

\item Objective function $g(V)$: As mentioned after \eqref{eq:opt0}, we use a $g(V)$ penalizing voltage violations, with parameter $\rho=100$. This function is convex and its gradient is Lipschitz continuous with parameter $L_g=100$. We use the voltage limits $V_{\max} = 1.06$ p.u. and $V_{\min} = 0.94$ p.u. as in \cite{hauswirth2017onlinePF}.
 
\item Feasible space $\mathcal{F}$: As mentioned after \eqref{eq:opt0}, we consider a feasible space limiting both the active and reactive power according to the profiles of distributed generation. This could easily be extended to more complex feasible space using the projection proposed in \cite{anese2016optimal}.

\item Descent rate: Experimentally, we have found that with a descent rate $\epsilon=0.001$ the proposed method is stable.
\end{itemize}

\subsection{Results}

\begin{figure}
\centering
\includegraphics[width=8.8cm]{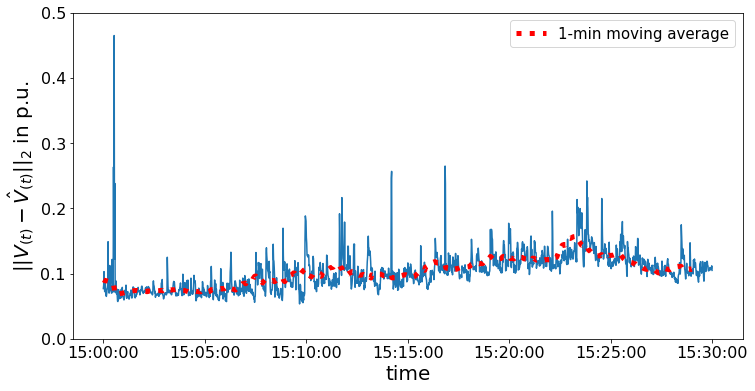}    
\caption{Euclidean norm at each time instant $(t)$ of the error between the estimated state $\hat{V}_{(t)}$ and the true state $V_{(t)}$.}
\label{fig:errorestim}
\end{figure}

\begin{figure}
\centering
\includegraphics[width=8.8cm]{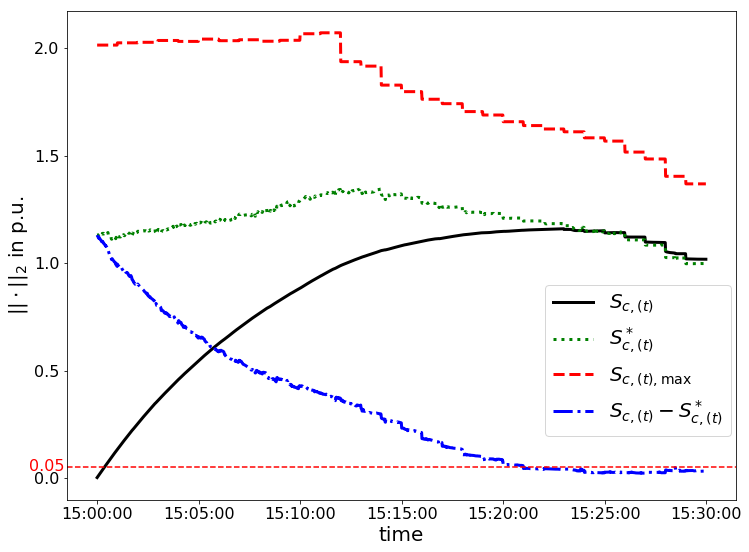}    
\caption{Euclidean norm at each time instant $(t)$ of the proposed set-points $S_{c,(t)}$, the optimum $S_{c,(t)}^*$, the maximum power available $S_{c,(t),\max}$ according to the power profiles represented in the feasible set $\mathcal{F}$, and the optimality error $S_{c,(t)}-S_{c,(t)}^*$. A line marks how the error stabilizes around $0.05$p.u..}
\label{fig:Scusederror}
\end{figure}

\begin{figure}
\centering
\includegraphics[width=8.8cm]{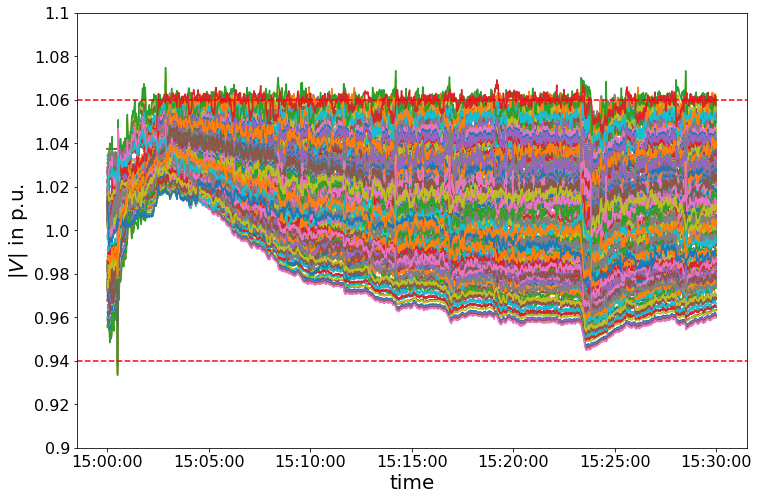}    
\caption{Magnitude at each time instant $(t)$ of the true voltage value $V_{(t)}$ for every node during the simulation time.}
\label{fig:nodesVtime}
\end{figure}

Given the theoretical result presented in Theorem \ref{thm:conv}, we monitor the optimization error norm $\normsz{S_{c,(t)}-S_{c,(t)}^*}$, and the estimation error norm $\normsz{\hat{V}_{(t)}-V_{(t)}^*}$. In Figure~\ref{fig:errorestim} and \ref{fig:Scusederror} we can observe how both the estimation error and the optimization errors decrease to a low value, less than $1\%$ for the estimation case, and then remain stable. These observations coincide with the results \eqref{eq:convboundedmsq} in Theorem \ref{thm:conv}, since the expected value of the norm is bounded, but does not necessarily converge to $0$. Despite that, it converges to close to $0$. This is a consequence of using fast time scales and thus a low covariance $\Sigma_l$, which in turn produces low constant bounds $C_V,C_S$.

The estimation error converges very quickly in a few seconds (iterations), while the optimization error needs about 20 min (1200 iterations). This supports the convergence rates of Theorem \ref{thm:conv}: the optimization error is slower due to the term $t$ multiplying the exponential part of the bound for the set-points case in \eqref{eq:convunbiassol} and \eqref{eq:convboundedmsq}.

From the curve of $\normsz{S_{c,(t)}^*}$ in Figure~\ref{fig:Scusederror}, it can be inferred that the optimum solution $S_{c,(t)}^*$ is indeed time-varying.
Nonetheless, it is remarkable that we observe an accurate convergence as predicted by Theorem \ref{thm:conv}, despite the fact that our theoretical assumption of zero mean process noise $\omega_{l,(t)}$ in \eqref{eq:dyneq} is not necessarily met for the true consumption data \cite{ECOdata} used in this simulation. This supports the statement in Remark \ref{rem:mean0} that for fast time-scales this drift can be neglected.

The feasible set $\mathcal{F}$ that we have used is not constant due to the time-varying profiles in solar radiation and wind speed \cite{solarprofile, windprofile}, see $\normsz{S_{c,(t),\max}}$ in Figure~\ref{fig:Scusederror}. However, this has not affected the convergence of our method, because the changes in $\mathcal{F}$ are also negligible at these fast time-scales, as mentioned in Remark \ref{rem:constfeas}.

Moreover, we obtain this accurate convergence despite the potential error due to the linear approximation done in \eqref{eq:PFapprox} and \eqref{eq:meas}. This is a consequence of using the real-time feedback optimization, since its robustness helps to correct potential model mismatches \cite{doerfler2019autonomousgrids}.

Finally, it can be observed in Figure~\ref{fig:nodesVtime} that the voltage magnitudes remain all the time within limits almost for every node. There are low violations \r{taking place at different nodes depending on the time step. These violations} may be due to the estimation uncertainty represented in the covariance matrix $P_{(t)}$ in \eqref{eq:kalmanest}, and the fact that we are not enforcing hard state constraints in \eqref{eq:opt}, but only penalising their violations. 

\section{Conclusions}\label{sec:conc}

In this paper we have proposed how to add a State Estimation (SE) to the online feedback optimization of the Optimal Power Flow (OPF), in order to control unmeasured states. We have formally proven that the interconnected system of SE and feedback optimization is stochastically stable and converges to the true state and optimal solution, respectively. Moreover, we have observed in a simulated test case how our method succeeds in driving the controllable elements towards near-optimal set-points while keeping an accurate state estimate.

Future work could include considering a nonlinear power flow and measurements equation; biased pseudo-measurements and process noise, to have non-zero mean time-varying loads; or using the uncertainty represented in the estimation error covariance matrix to increase the voltage restrictions and achieve fewer and smaller voltage violations.

\bibliographystyle{IEEEtran}
\bibliography{IEEEabrv,ifacconf}

\appendices

\section{Proof of Theorem \ref{thm:conv}}\label{app:thmproof}

\begin{mdframed}[hidealllines=true, backgroundcolor=white, 
innerleftmargin=3pt, innerrightmargin=3pt, leftmargin=-3pt, rightmargin=-3pt]

For the proof, we first write the whole closed-loop system in Figure~\ref{fig:bldiag} connecting together all the parts in \eqref{eq:meas}, \eqref{eq:dyneq}, \eqref{eq:kalmanest}, \eqref{eq:projgrad}. Denoting the estimation error as $e_{(t)}=\hat{V}_{(t)}-V_{(t)}$, we define the state $x_{(t)}$, output $z_{(t)}$, and disturbance $\omega_{(t)}$ as

\begin{equation}\label{eq:vari}\arraycolsep=1pt
\begin{array}{c}
x_{(t)} = \left[\begin{array}{c} V_{(t)} \\ e_{(t)} \\ S_{c,(t)} \end{array}\right], 
z_{(t)} = \left[\begin{array}{c} \hat{V}_{(t)}\\ S_{c,(t)} \end{array}\right], \;
\omega_{(t)} = \left[\begin{array}{c} \omega_{l,(t)}\\ \omega_{y,(t)} \end{array}\right]
\end{array}
\end{equation} 

Next we define the control input $u_{(t+1)} = \phi_{\epsilon} (z_{(t)})$, depending on the output $z_{(t)}$ through the nonlinear feedback $\phi_{\epsilon}(\cdot)$. This operator $\phi_{\epsilon}(\cdot)$ is the generalized gradient mapping \cite{fazlyab2018analysis} adapted to the projected gradient descent in \eqref{eq:projgrad}, \r{and it is a mere reformulation of \eqref{eq:projgrad}}:

\begin{equation}
\phi_{\epsilon}(z) = \frac{1}{\epsilon} \big(S_c - \Pi_{\mathcal{F}} \big[ S_{c} - \epsilon \big(\nabla_{S_c} f(S_{c}) + B_c^T \nabla_V g(\hat{V}) \big) \big]\big)
\end{equation}

As a result, we get the nonlinear stochastic closed-loop interconnected system represented in Figure~\ref{fig:bldiag} and expressed as
\begin{equation}\label{eq:intsys}\arraycolsep=1pt\begin{array}{rl}
x_{(t+1)} =  &
\left[\begin{array}{ccc} I_d & 0 & 0 \\ 0 & I_d-K_{(t+1)}H & 0 \\ 0 & 0 & I_d \end{array}\right] x_{(t)} +
\left[\begin{array}{cc} -\epsilon B_l  \\ 0  \\ -\epsilon I_d \end{array}\right] u_{(t+1)} \\[0.5cm]
& +
\left[\begin{array}{cc} B_l & 0 \\ -(I_d-K_{(t+1)}H)B_l & K_{(t+1)}  \\ 0 & 0 \end{array}\right] \omega_{(t+1)} \\[0.5cm]
z_{(t)} = & \left[\begin{array}{ccc} I_d & I_d & 0 \\ 0 & 0 & I_d \end{array}\right] x_{(t)} \\
u_{(t+1)} = & \phi_{\epsilon} (z_{(t)}).
\end{array}
\end{equation}
\end{mdframed}

Considering a perfect estimation $e_{(t)}^*=0$, and using the instantaneous and expected global optimal values defined before Theorem \ref{thm:conv}, $S_{c,(t)}^*$ and $S_{c}^{\mathbb{E}*}$ respectively, we can define the desired points $x_{(t)}^* = [(V_{(t)}^*)^T,0,(S_{c,(t)}^*)^T]^T$ and $x^{\mathbb{E}*} = [(V^{\mathbb{E}*})^T,0,(S_{c}^{\mathbb{E}*})^T]^T$. Due to optimality, these points are instantaneous fixed points of the projected gradient descent. Then, at these points $\phi_{\epsilon}(z_{(t)}^*)=0, \phi_{\epsilon}(z^{\mathbb{E}*})=0$, and thus $x^{\mathbb{E}*}$ is an equilibrium point, and $x_{(t)}^*$ are instantaneous equilibrium points at each time $(t)$. Next, we analyse the stability of these points to prove the convergence of \eqref{eq:kalmanest} and \eqref{eq:projgrad}.

\subsection{Proof of \eqref{eq:convboundedmsq}}

We start formulating some results:
\r{
\begin{lemma}\label{lemma:LipConst}
The objective as a function of $S_c$: $\tilde{f}(S_c)=f(S_c)+g(V(S_c))=f(S_c)+g(V_0+B_cS_c+B_lS_l)$ is $\eta$-strongly convex and has a Lipschitz continuous gradient with parameter $L=L_f+\normsz{B_c}^2 L_g$ \r{(see Assumption \ref{ass:thm} for $L_f,L_g,\eta$).}.
\end{lemma}
\begin{proof}
\begin{enumerate}[leftmargin=*]
\item Strongly convex: Since $g(V)$ is convex on $V$, $g(V_0+B_cS_c+B_lS_l)$ is convex on $S_c$. Then, since $f(S_c)$ is $\eta$-strongly convex, so is $\tilde{f}(S_c)$. 
\item Lipschitz: We have $\nabla_{S_c} \tilde{f}(S_c) = \nabla_{S_c} f(S_c) + B_c^T \nabla_V g(V_0 + B_cS_c + B_lS_l)$, then
\begin{equation}\begin{array}{l}
\normsz{\nabla_{S_c} \tilde{f}(S_{c,1}) - \nabla_{S_c} \tilde{f}(S_{c,2})} \\
\leq L_f\normsz{S_{c,1}-S_{c,2}}+\normsz{B_c} L_g \normsz{B_c(S_{c,1}-S_{c,2})} \\
\leq (L_f+\normsz{B_c}^2 L_g) \normsz{S_{c,1}-S_{c,2}}
\end{array}
\end{equation}
\end{enumerate}
\end{proof}
}

\begin{lemma}\label{lemma:lipOp}
There exists a Lipschitz constant $L_{S,\text{opt}}$, such that the operators mapping the subsequent values $S_{l,(t)},S_{l,(t+1)}$, to the respective optimal solutions $S_{c,(t)}^*,S_{c,(t+1)}^*$ of \eqref{eq:opt} satisfy: 
\begin{equation}\arraycolsep=1pt\begin{array}{rl}
\normsz{S_{c,(t+1)}^*-S_{c,(t)}^*} 
& \leq L_{S,\text{opt}} \normsz{B_l(S_{l,(t+1)}-S_{l,(t)})} \\
& = L_{S,\text{opt}} L_{S,\text{opt}} \normsz{\omega_{l,(t+1)}}  
\end{array}
\end{equation}
\end{lemma}
\begin{proof}
The proof is in Appendix \ref{app:lipOp}.
\end{proof}

\r{This lemma goes in accordance with the statement in \cite[(14)]{simonetto2017timevaropt}, limiting the change of optimal solutions to the size of the time step, or in our case, the change in the load inputs in this time step.}

\begin{mdframed}[hidealllines=true,backgroundcolor=white, 
innerleftmargin=3pt,innerrightmargin=3pt,leftmargin=-3pt,rightmargin=-3pt]


\begin{thm}\label{thm:lyapconv}\cite[Theorem~2]{tarn1976observers}
If there is a stochastic process $z_{(t)}$, a function $\mathcal{V}_{(t)}(\cdot)$ and real numbers $\check{\nu},\hat{\nu},\zeta>0$ and $0<\delta\leq 1$ such that
\begin{subequations}\label{eq:lyapcond}
\begin{align}\arraycolsep=1pt
&\begin{array}{l}
\check{\nu} \normsz{z_{(t)}}^2 \leq \mathcal{V}_{(t)}(z_{(t)}) \leq \hat{\nu} \normsz{z_{(t)}}^2
\end{array}\label{eq:lyapcond1}
\\
&\begin{array}{l}
\expect{ \mathcal{V}_{(t+1)}(z_{(t+1)}) | z_{(t)}} \leq \zeta + (1-\delta) \mathcal{V}_{(t)}(z_{(t)})
\end{array}\label{eq:lyapcond2}
\end{align}
\end{subequations}
then 
\begin{equation*}\arraycolsep=1pt\begin{array}{rl}
\expect{\normsz{z_{(t)}}^2} & \leq \frac{\hat{\nu}}{\check{\nu}} \expect{\normsz{z_{(0)}}^2} (1-\delta)^t + \frac{\zeta}{\check{\nu}} \sum_{i=1}^{t-1} (1-\delta)^i \\[0.2cm]
& \leq \frac{\hat{\nu}}{\check{\nu}} \expect{\normsz{z_{(0)}}^2} (1-\delta)^t + \frac{\zeta}{\check{\nu}\delta} ,
\end{array}
\end{equation*}
\end{thm}
\end{mdframed}

\r{
According to Assumption \ref{ass:thm}, $H$ has full-column rank, and $\Sigma_y,\Sigma_l$ full rank. Then the system is uniformly completely observable \cite[Chapter~7]{jazwinski1970stochfilt}. Thus, the covariance matrix $P_{(t)}$ is also upper for all $(t)$ \cite[Lemma~7.1]{jazwinski1970stochfilt}. So apart from the lower bound $\check{\sigma}$ in Assumption \ref{ass:thm}, there exist an upper bound $\hat{\sigma} >0$ on $P_{(t)}$, depending on $\Sigma_l,\Sigma_y,H,B_c,B_l$, so that: $\check{\sigma} I_d \preceq P_{(t)} \preceq	\hat{\sigma} I_d, \; \forall t$. As a result, $K_{(t)}$ is also bounded.
}

Despite the non-linear feedback, since we have a linear open-loop system \eqref{eq:dyneq} and estimator \eqref{eq:kalmanest}, the separation principle \cite{aastrom2012introduction} holds and we can analyse the stability and convergence of the estimator and controller separately. \r{This separation can be directly observed in \eqref{eq:intsys}, since there is a $0$ multiplying the input $u(t+1)$ for the index corresponding to the estimation error $e(t)$.} Then, using the \textit{Lyapunov} function $\mathcal{V}_{(t)}(e_{(t)}) = e_{(t)}^T P_{(t)}^{-1} e_{(t)}$, we can prove the results of Theorem \ref{thm:conv} for the estimation error $e(t)$: The first condition \eqref{eq:lyapcond1} is satisfied using the bounds of $P_{(t)}$, and \eqref{eq:lyapcond2} can be proven using a similar reasoning as in \cite{Reif1999kalmanstab}:

\begin{mdframed}[hidealllines=true,backgroundcolor=white, 
innerleftmargin=3pt,innerrightmargin=3pt,leftmargin=-3pt,rightmargin=-3pt]

\begin{equation}\label{eq:lyapestimerror}\arraycolsep=1pt\begin{array}{rl}
& e_{(t+1)}^T P_{(t+1)}^{-1} e_{(t+1)} \vspace{-0.1cm}\\

\stackrel{a)}{=} & e_{(t)}^T (I_d-K_{(t+1)}H)^T P_{(t+1)}^{-1} (I_d-K_{(t+1)}H) e_{(t)}\\
	& + \omega_{y,(t+1)}^T K_{(t+1)}^T P_{(t+1)}^{-1} K_{(t+1)} \omega_{y,(t+1)} \\
	& + \omega_{l,(t+1)}^T (I_d-K_{(t+1)}H)^T P_{(t+1)}^{-1} (I_d-K_{(t+1)}H) \omega_{l,(t+1)} \\
	& + (\text{crossed terms on } e_{(t)},\omega_{y,(t)},\omega_{l,(t)}) \vspace{-0.1cm}\\
	
\stackrel{b)}{=} & e_{(t)}^T (P_{(t)}+\Sigma_l)^{-1} (I_d-K_{(t+1)}H) e_{(t)}\\
	& + \omega_{y,(t+1)}^T K_{(t+1)}^T P_{(t+1)}^{-1} K_{(t+1)} \omega_{y,(t+1)} \\
	& + \omega_{l,(t+1)}^T (I_d-K_{(t+1)}H)^T P_{(t+1)}^{-1} (I_d-K_{(t+1)}H) \omega_{l,(t+1)} \\
	& + (\text{crossed terms on } e_{(t)},\omega_{y,(t)},\omega_{l,(t)}) \vspace{-0.1cm}\\	

\stackrel{c)}{\leq} & e_{(t)}^T P_{(t)}^{-1} e_{(t)}\\
	& - e_{(t)}^T H^T(H(P_{(t)}+\Sigma_l)H^T+\Sigma_y)^{-1}H e_{(t)} \\
	& + \omega_{y,(t+1)}^T K_{(t+1)}^T P_{(t+1)}^{-1} K_{(t+1)} \omega_{y,(t+1)} \\
	& + \omega_{l,(t+1)}^T (I_d-K_{(t+1)}H)^T P_{(t+1)}^{-1} (I_d-K_{(t+1)}H) \omega_{l,(t+1)} \\
	& + (\text{crossed terms on } e_{(t)},\omega_{y,(t)},\omega_{l,(t)}) \vspace{-0.1cm}\\	

\stackrel{d)}{<} & e_{(t)}^T P_{(t)}^{-1} e_{(t)} - \psi \normsz{e_{(t)}}^2 \\
	& + \frac{1}{\check{\sigma}} (\normsz{K_{(t+1)}\omega_{y,(t)}}^2+\normsz{(I_d-K_{(t+1)}H)\omega_{l,(t)}}^2) \\
	& + (\text{crossed terms on } e_{(t)},\omega_{y,(t)},\omega_{l,(t)}) \vspace{-0.1cm}\\	

\stackrel{e)}{=} & (1-\psi\check{\sigma} ) e_{(t)}^T P_{(t)}^{-1} e_{(t)}\\
	& + \frac{1}{\check{\sigma}} (\normsz{K_{(t+1)}\omega_{y,(t)}}^2+\normsz{(I_d-K_{(t+1)}H)\omega_{l,(t)}}^2) \\
	& + (\text{crossed terms on } e_{(t)},\omega_{y,(t)},\omega_{l,(t)}) \\		

\end{array}
\end{equation}
where in $a)$ we have expanded $e_{(t+1)}$ using \eqref{eq:intsys}; in $b)$ we have used the expression of $P_{(t+1)}$ in \eqref{eq:kalmanest}; in $c)$ we use the expression of $K_{(t+1)}$ in \eqref{eq:kalmanest}, and that $\Sigma_l \succeq 0$, and thus $P_{(t)}^{-1} \succeq (P_{(t)}+\Sigma_l)^{-1}$; in $d)$ we use the upper bound of $P_{(t+1)}^{-1}$, and that since $H(P_{(t)}+\Sigma_l)H^T \succeq 0$, $ \Sigma_y \succ 0$ and $H$ is full-column rank, then $H^T(H(P_{(t)}+\Sigma_l)H^T+\Sigma_y)^{-1}H \succ 0$, so there exists $\psi>0$ such that $d)$ is true with:
\begin{equation}\label{eq:psi}\arraycolsep=1pt\begin{array}{rl}
\psi < & \lambda_{\min} \big(H^T(H(P_{(t)}+\Sigma_l)H^T+\Sigma_y)^{-1}H \big) \\ 
\end{array}
\end{equation}
where $\lambda_{\min}(\cdot)$ denote the minimum eigenvalue; and in $e)$ we use the lower bound $\check{\sigma}$ of $P_{(t)}$. Note \r{that \eqref{eq:psi} provides and upper bound on $\psi$,} so we can choose $\psi$ such that $\psi\check{\sigma} <1$. Taking the expectation $\mathbb{E}[\cdot]$, since $\expect{\omega_{y,(t)}}=\expect{\omega_{l,(t)}}=0$, the crossed terms vanish and we get:
\begin{equation}\label{eq:estimconv}\begin{array}{l}
\expect{e_{(t+1)}^T P_{(t+1)}^{-1} e_{(t+1)} \mid e_{(t)}} \\
\leq (1-\psi\check{\sigma} ) e_{(t)}^T P_{(t)}^{-1} e_{(t)} + \tau (\text{trace}(\Sigma_y)+\text{trace}(\Sigma_l))
\end{array}
\end{equation} 
where we have used $\expect{\normsz{\omega_{(\cdot),(t)}}^2}=\text{trace}(\Sigma_{(\cdot)})$, and $\tau = \frac{1}{\check{\sigma}}\max(\normsz{K_{(t+1)}}^2,\normsz{I_d-K_{(t+1)}H}^2) < \infty$. Then, we can conclude that $e_{(t)}$ is exponentially bounded in mean square and provide the corresponding bound \eqref{eq:convboundedmsq} using Theorem \ref{thm:lyapconv}:
\begin{equation}\label{eq:expboundestim}\begin{array}{l}
\expect{\normsz{e_{(t)}}^2} \\ 
\leq \frac{\hat{\sigma}}{\check{\sigma}}\big(1-\psi\check{\sigma} \big)^{t} \expect{\normsz{e_{(0)}}^2} 
+ \frac{\tau}{\psi} (\text{trace}(\Sigma_y)+\text{trace}(\Sigma_l)) \\
\stackrel{t \rightarrow \infty}{\rightarrow} \frac{\tau}{\psi} (\text{trace}(\Sigma_y)+\text{trace}(\Sigma_l))
\end{array}
\end{equation} 
\end{mdframed}

Next we study at the convergence and stability of $S_{c,(t)}$. Due to Lemma \ref{lemma:lipOp}, we have
\r{
\begin{equation}\label{eq:timeconv}\arraycolsep=1pt\begin{array}{l}
\normsz{S_{c,(t+1)}-S_{c,(t+1)}^*} \\
\leq \normsz{S_{c,(t+1)}-S_{c,(t)}^*} + \normsz{S_{c,(t+1)}^*-S_{c,(t)}^*}  \\ 
\leq \normsz{S_{c,(t+1)}-S_{c,(t)}^*} + L_{S,\text{opt}} \normsz{\omega_{l,(t+1)}} 
\end{array}
\end{equation}
}

Then we have
\begin{equation}\label{eq:optpastconv}\arraycolsep=1pt\begin{array}{rl}
& \normsz{S_{c,(t+1)}-S_{c,(t)}^*} \vspace{-0.1cm} \\

\stackrel{a)}{=} & \normsz{\Pi_{\mathcal{F}} \big[ S_{c,(t)} 
- \epsilon (\nabla_{S_c} f(S_{c,(t)}) + B_c^T \nabla_V g(\hat{V}_{(t)}) ) \big] \\
& - \Pi_{\mathcal{F}} \big[S_{c,(t)}^* - \epsilon (\nabla_{S_c} f(S_{c,(t)}^*) + B_c^T \nabla_V g(V_{(t)}^*) ) \big] } \vspace{-0.15cm} \\

\stackrel{b)}{\leq} & \normsz{ S_{c,(t)} 
- \epsilon (\nabla_{S_c} f(S_{c,(t)}) + B_c^T \nabla_V g(\hat{V}_{(t)}) )  \\
& - S_{c,(t)}^* 
+ \epsilon (\nabla_{S_c} f(S_{c,(t)}^*) + B_c^T \nabla_V g(V_{(t)}^*)) } \vspace{-0.15cm}\\

\stackrel{c)}{\leq} & \normsz{ S_{c,(t)} - S_{c,(t)}^* 
- \epsilon \big(\nabla_{S_c} f(S_{c,(t)}) - \nabla_{S_c} f(S_{c,(t)}^*) \\
& + B_c^T \nabla_V g(V_{(t)}) - B_c^T \nabla_V g(V^*_{(t)}) \big) } \\
& + \normsz{\epsilon B_c^T ( \nabla_V g(\hat{V}_{(t)}) - \nabla_V g(V_{(t)})) } \vspace{-0.1cm}\\

\stackrel{d)}{\leq} & \sqrt{1-2\eta\epsilon+\epsilon^2 L^2} \normsz{ S_{c,(t)} - S_{c,(t)}^*} + \epsilon \normsz{B_c} L_g \normsz{e_{(t)}} \\

\end{array}
\end{equation}
where in $a)$ we use that $S_{c,(t)}^*$ is a fixed point of the projected gradient descent \eqref{eq:projgrad} due to its optimality; in $b)$ we remove the projection; in $c)$ we add and subtract $\epsilon B_c^T \nabla_V g(V_{(t)})$, and then use the triangle inequality; in $d)$, for the first norm we use the strong convexity and the Lipschitz continuity given by Lemma \ref{lemma:LipConst}; and for the second norm the Lipschitz continuity of $\nabla_V g(V)$. Defining $r(\epsilon) = \sqrt{1-2\eta\epsilon+\epsilon^2 L^2}$, we combine \eqref{eq:timeconv} and \eqref{eq:optpastconv}, and take expectations to get:

\begin{equation}\label{eq:expboundoptim}\arraycolsep=1pt 
\begin{array}{l}
\begin{array}{rl}
& \expect{\normsz{S_{c,(t+1)}-S_{c,(t+1)}^*}} \\ 
\leq &  r(\epsilon) \expect{\normsz{S_{c,(t)}-S_{c,(t)}^*}} \\
&  + \epsilon \normsz{B_c} L_g \expect{\normsz{e_{(t)}}} + L_{S,\text{opt}} \expect{\normsz{\omega_{l,(t)}}} \vspace{-0.1cm} \\

\stackrel{a)}{\leq} & r(\epsilon)^{t} \expect{\normsz{S_{c,(0)}-S_{c,(0)}^*}} \\
& + \sum_{k=0}^{t} r(\epsilon)^{(t-k)} \big(L_{S,\text{opt}} \sqrt{ \text{trace}(\Sigma_l)} \\ 
& + \epsilon \normsz{B_c} L_g \sqrt{\expect{\normsz{e_{(k)} }^2}} \big) \vspace{-0.1cm} \\ 


\stackrel{b)}{\leq} & r(\epsilon)^{t} \expect{\normsz{S_{c,(0)}-S_{c,(0)}^*}} \\
& + \sum_{k=0}^{t} r(\epsilon)^{(t-k)} \Big(  L_{S,\text{opt}} \sqrt{\text{trace}(\Sigma_l)} \\
& + \epsilon \normsz{B_c} L_g  \sqrt{ \frac{\hat{\sigma}}{\check{\sigma}}\big(1-\psi\check{\sigma} \big)^{k} \expect{\normsz{e_{(0)}}^2} } \\
& +  \epsilon \normsz{B_c} L_g \sqrt{ \frac{\tau}{\psi} (\text{trace}(\Sigma_y)+\text{trace}(\Sigma_l)) } \Big) \vspace{-0.15cm} \\

\stackrel{c)}{\leq} & r(\epsilon)^{t} \expect{\normsz{S_{c,(0)}-S_{c,(0)}^*}} \\
& + \epsilon \normsz{B_c} L_g \sqrt{\frac{\hat{\sigma}}{\check{\sigma}} \expect{\normsz{e_{(0)}}^2}} t \max\big(r(\epsilon),\sqrt{1-\psi\check{\sigma}}\big)^t \\
& + \frac{1}{1-r(\epsilon)} \Big(L_{S,\text{opt}} \sqrt{\text{trace}(\Sigma_l)} \\
& + \epsilon \normsz{B_c} L_g \sqrt{\frac{\tau}{\psi} (\text{trace}(\Sigma_y)+\text{trace}(\Sigma_l))} \Big) \\

\end{array} \\

\stackrel{t \rightarrow \infty}{\rightarrow}  \frac{1}{1-r(\epsilon)} \Big(L_{S,\text{opt}} \sqrt{\text{trace}(\Sigma_l)} \\
\hspace{0.8cm} + \epsilon \normsz{B_c} L_g \sqrt{\frac{\tau}{\psi} (\text{trace}(\Sigma_y)+\text{trace}(\Sigma_l))} \Big)
\end{array}
\end{equation} 
where in $a)$ we apply the inequality in \eqref{eq:optpastconv} $(t)$-times, and use the Jensen inequality for the concave function $\sqrt{(\cdot)}$ on $\expect{\normsz{\omega_{l,(t)}}}$ and $\expect{\normsz{e_{(k)}}}$: $\expect{\normsz{(\cdot)}} \leq \sqrt{\expect{\normsz{(\cdot)}^2}}$; in $b)$ we substitute $\expect{\normsz{e_{(k)}}^2} $ using \eqref{eq:expboundestim} and use that $\sqrt{(\cdot)}$ is subadditive: $\sqrt{x+y}\leq \sqrt{x}+\sqrt{y}, \; \forall x,y >0$; and in $c)$ \r{and the limit, we sum over $k$ and use the following: $r(\epsilon)<1$ if choosing a step size $\epsilon < \frac{2\eta}{(L_f+\normsz{B_c}^2 L_g)^2+\normsz{B_c}^2 L_g^2}$ according to Theorem \ref{thm:conv}; as explained after \eqref{eq:psi}, $\psi$ can be chosen so that $1-\psi\check{\sigma}<1$; $r(\epsilon)>1-\frac{\eta^2}{L^2}>0$, $1-\psi\check{\sigma}>0$; then $0<r(\epsilon)^{t-k}\sqrt{1-\psi\check{\sigma}}^k \leq \max\big(r(\epsilon),\sqrt{1-\psi\check{\sigma}}\big)^t$, since $t \geq k \geq 0$, with $\max\big(r(\epsilon),\sqrt{1-\psi\check{\sigma}}\big)<1$.}

\subsection{Proof of \eqref{eq:convunbiassol}}

Taking expectations on the interconnected system \eqref{eq:intsys}, we can eliminate the disturbance noise $\omega_{(t)}$. Without it and using the same procedure as before for the estimation part, we get
\begin{equation}\label{eq:biasconv}\begin{array}{l}
\expect{e_{(t+1)}\mid e_{(t)}}^T P_{(t+1)}^{-1} \expect{e_{(t+1)} \mid e_{(t)}} \\
\leq (1-\psi\check{\sigma} ) e_{(t)}^T P_{(t)}^{-1} e_{(t)} 
\end{array}
\end{equation} 
so that instead of \eqref{eq:estimconv} we have
\begin{equation}\label{eq:biasestimconv}\arraycolsep=1pt\begin{array}{l}
\normsz{\expect{e_{(t)}}}^2 
\leq \frac{\hat{\sigma}}{\check{\sigma}}\big(1-\psi\check{\sigma} \big)^{t} \normsz{\expect{e_{(0)}}}^2
\stackrel{t \rightarrow \infty}{\rightarrow} 0,
\end{array}
\end{equation} 
and for the optimality part we have
\begin{equation}\label{eq:biasoptimconv}\arraycolsep=1pt
\begin{array}{l}
\begin{array}{rl}
& \normsz{S_{c,(t+1)} - S_{c}^{\mathbb{E}*}} \vspace{-0.1cm}\\

\stackrel{a)}{\leq} & r(\epsilon) \normsz{ S_{c,(t)} - S_{c}^{\mathbb{E}*}} + \epsilon \normsz{B_c} L_g \normsz{\expect{e_{(t)}}} \\

\stackrel{b)}{\leq} & r(\epsilon) \normsz{ S_{c,(t)} - S_{c}^{\mathbb{E}*}} + \epsilon \normsz{B_c} L_g \sqrt{\frac{\hat{\sigma}}{\check{\sigma}}\normsz{\expect{e_{(0)}}}^2} \sqrt{1-\psi\check{\sigma}}^{t} \vspace{-0.1cm}\\

\stackrel{c)}{\leq} & r(\epsilon)^t \normsz{ S_{c,(0)} - S_{c}^{\mathbb{E}*}} \\ 
& + \epsilon \normsz{B_c} L_g \sqrt{ \frac{\hat{\sigma}}{\check{\sigma}} \normsz{\expect{e_{(0)}}}^2} t \max\big(r(\epsilon),\sqrt{1-\psi\check{\sigma}}\big)^t  \\

\end{array} \\

\stackrel{t \rightarrow \infty}{\rightarrow} 0

\end{array}
\end{equation}
where in $a)$ we use same arguments as in \eqref{eq:optpastconv}, but adding and subtracting $\epsilon B_c^T \nabla_V g(\expect{V_{(t)}})$; in $b)$ we substitute $\normsz{\expect{e_{(t)}}}^2$ using \eqref{eq:biasestimconv}; in $c)$ we apply the previous inequality $(t)$-times and sum the terms as in \eqref{eq:expboundoptim}

\subsection{Proof of Lemma \ref{lemma:lipOp}}\label{app:lipOp}

\begin{equation}\label{eq:lemma}\arraycolsep=1pt
\begin{array}{l}
\begin{array}{rl}
& \normsz{S_{c,(t+1)}^*-S_{c,(t)}^*} \vspace{-0.1cm}\\

\stackrel{a)}{=} & \normsz{\Pi_{\mathcal{F}} \big[ S_{c,(t+1)}^* 
- \epsilon \big(\nabla_{S_c} f(S_{c,(t+1)}^*) + B_c^T \nabla_V g(V_{(t+1)}^*) \big) \big] \\
& - \Pi_{\mathcal{F}} \big[S_{c,(t)}^* - \epsilon \big(\nabla_{S_c} f(S_{c,(t)}^*) + B_c^T \nabla_V g(V_{(t)}^*) \big) \big] } \vspace{-0.15cm}\\

\stackrel{b)}{\leq} & \normsz{S_{c,(t+1)}^* 
- \epsilon \big(\nabla_{S_c} f(S_{c,(t+1)}^*) + B_c^T \nabla_V g(V_{(t+1)}^*) \big) \\
& - S_{c,(t)}^* + \epsilon \big(\nabla_{S_c} f(S_{c,(t)}^*) + B_c^T \nabla_V g(V_{(t)}^*) \big) } \vspace{-0.15cm} \\

\stackrel{c)}{\leq} & \normsz{S_{c,(t+1)}^* - S_{c,(t)}^*
- \epsilon \big(\nabla_{S_c} f(S_{c,(t+1)}^*) - \nabla_{S_c} f(S_{c,(t)}^*) \\
& + B_c^T \nabla_V g(B_cS_{c,(t+1)}^* + B_lS_{l,(t)}) \\ 
& - B_c^T \nabla_V g(B_cS_{c,(t)}^* + B_lS_{l,(t)}) \big) } \\
& + \normsz{\epsilon B_c^T (\nabla_V g(B_cS_{c,(t+1)}^* + B_lS_{l,(t+1)}) \\
& - \nabla_V g(B_cS_{c,(t+1)}^* + B_lS_{l,(t)}))} \vspace{-0.20cm} \\

\stackrel{d)}{\leq} & r(\epsilon) \normsz{S_{c,(t+1)}^* - S_{c,(t)}^*} \\
& + \epsilon \normsz{B_c} L_g \normsz{B_l(S_{l,(t+1)}-S_{l,(t)})} \vspace{-0.15cm}\\

\stackrel{e)}{\leq} & r(\epsilon)^k \normsz{S_{c,(t+1)}^*-S_{c,(t)}^*} \\
& + \frac{1-r(\epsilon)^{k}}{1-r(\epsilon)} \epsilon \normsz{B_c} L_g  \normsz{B_l(S_{l,(t+1)}-S_{l,(t)})} \\
\end{array} \\

\stackrel{k \rightarrow \infty}{\rightarrow} \frac{\epsilon}{1-r(\epsilon)}  \normsz{B_c} L_g \normsz{B_l(S_{l,(t+1)}-S_{l,(t)})} \\

\end{array}
\end{equation}
where in $a)$ we use that both solutions are fixed points of the projected gradient descent \eqref{eq:projgrad} due to their optimality; in $b)$ we remove the projection; in $c)$ we use \eqref{eq:PFapprox}: $V_{(t)}^*=B_c S_{c,(t)}^* + B_l S_{l,(t)}$, we add and subtract $\epsilon B_c^T \nabla_V g(B_cS_{c,(t+1)}^* + B_lS_{l,(t)}$ and use the triangle inequality; in $d)$ we use the strong convexity and the gradient Lipschitz continuity as in \eqref{eq:optpastconv} defining $r(\epsilon) = \sqrt{1-2\eta\epsilon+\epsilon^2 L^2}$; in $e)$ we apply the previous inequality $k$ times; \r{and in the limit we use that $\abs{r(\epsilon)}<1$ as explained after \eqref{eq:expboundoptim}}. 

\r{Since \eqref{eq:lemma} is true for all $\epsilon$ such that $r(\epsilon)<1$, we can choose $L_{S,\text{opt}}=\min_{\{\epsilon \mid r(\epsilon)<1\}} \frac{\epsilon}{1-r(\epsilon)}$. However, note that since $ \lim_{\epsilon \rightarrow 0} \frac{\epsilon}{1-r(\epsilon)} =\frac{1}{2\eta}$, so the Lipschitz constant $L_{S,\text{opt}}$ is not arbitrarily small.}

\end{document}